\documentclass{article}
\usepackage{ijcai25}

\usepackage{csquotes}
\usepackage{times}
\usepackage{soul}
\usepackage{url}
\usepackage[hidelinks]{hyperref}
\usepackage[utf8]{inputenc}
\usepackage[small]{caption}
\usepackage{graphicx}
\usepackage{amsmath}
\usepackage{amsthm}
\usepackage{booktabs}
\usepackage{algorithm}
\usepackage{algorithmic}
\usepackage[switch]{lineno}
\usepackage{tikz}
\usepackage{comment}
\usepackage{subcaption}
\usepackage{amsmath}
\usepackage{amssymb}
\usepackage{todonotes}
\usepackage{pgfplots}
\usepackage{todonotes}
\usepackage{stmaryrd}
\usetikzlibrary{patterns,calc}

\newtheorem{theorem}{Theorem}
\newtheorem{problem}{Problem}
\newtheorem{lemma}{Lemma}

\title{Learning Probabilistic Temporal Logic Specifications for Stochastic Systems}

\author{
Rajarshi Roy$^1$
\and
Yash Pote$^2$\and
David Parker$^{1}$\And
Marta Kwiatkowska$^1$\\
\affiliations
$^1$University of Oxford, Oxford OX1 2JD, UK\\
$^2$National University of Singapore\\
\emails
\{rajarshi.roy, david.parker, marta.kwiatkowska\}@cs.ox.ac.uk,yashppote@gmail.com
}

\usepackage{stmaryrd}
\usepackage{pifont}
\usepackage{fontawesome}
\usepackage{amssymb}

\newcommand{\fanBr}[1]{\langle\!\langle#1\rangle\!\rangle\!}

\newcommand{\nat}{\mathbb{N}}

\newcommand{\LTL}{\textrm{LTL}}

\newcommand{\PLTL}{\textrm{PLTL}}
\newcommand{\PLTLB}{\textrm{PLTL}$^{+}$}

\newcommand{\agent}{\protect$\triangle$}
\newcommand{\decor}{\protect{\Large$\ast$}}
\newcommand{\office}{\protect{\huge$\circ$}}
\newcommand{\coffee}{\protect{\small\faCoffee}}

\newcommand{\markov}{M}
\newcommand{\PM}{\mathrm{Pr}}

\newcommand{\paths}{\Pi}

\newcommand{\prop}{\mathrm{AP}}

\newcommand{\literals}{\Lambda}

\DeclareMathOperator{\lF}{\mathbf{F}}
\DeclareMathOperator{\lG}{\mathbf{G}}
\DeclareMathOperator{\lU}{\mathbf{U}}
\DeclareMathOperator{\lX}{\mathbf{X}}

\DeclareMathOperator{\lfalse}{\mathit{false}}
\DeclareMathOperator{\ltrue}{\mathit{true}}

\DeclareMathOperator{\lP}{\mathbf{P}\!}

\newcommand{\sample}{\mathcal{S}}

\newcommand{\formulalist}{\mathcal{F}}
\newcommand{\discardlist}{\mathcal{D}}
\newcommand{\boollist}{\mathcal{B}}
\newcommand{\op}{\circ}
\newcommand{\size}[1]{|#1|}
\newcommand{\depth}[1]{d(#1)}
\newcommand{\automata}{\mathcal{A}}
\newcommand{\probvector}{V}
\newcommand{\score}{\sigma}

\newcommand{\knowA}{\mathtt{kA}}
\newcommand{\knowB}{\mathtt{kB}}

\newcommand{\tool}{\texttt{PriTL}}
\newcommand{\heap}{\mathcal{H}}

\begin{document}

\maketitle

\begin{abstract}
There has been substantial progress in the inference of
formal behavioural specifications from sample trajectories,
for example using Linear Temporal Logic (\LTL{}).
However, these techniques cannot handle specifications that correctly
characterise systems with stochastic behaviour,
which occur commonly in reinforcement learning and formal verification.
We consider the passive learning problem of inferring a Boolean combination
of probabilistic LTL (\PLTL{}) formulas from a set of Markov chains,
classified as either positive or negative.
We propose a novel learning algorithm that infers concise \PLTL{} specifications,
leveraging grammar-based enumeration, search heuristics, probabilistic model checking and Boolean set-cover procedures.
We demonstrate the effectiveness of our algorithm in two use cases: learning from policies induced by RL algorithms and learning from variants of a probabilistic model.
In both cases, our method automatically and efficiently extracts \PLTL{} specifications that succinctly characterize the temporal differences between the policies or model variants.

\end{abstract}


\section{Introduction}

Temporal logic is a powerful formalism used not only for writing correctness specifications in formal methods but also for defining non-Markovian goals and objectives in reinforcement learning (RL) and control tasks~\cite{DBLP:conf/iros/LiVB17,DBLP:conf/ijcai/CamachoIKVM19,DBLP:conf/cdc/HasanbeigKAKPL19,DBLP:conf/icra/Bozkurt0ZP20}.
Among temporal logics, Linear Temporal Logic (\LTL{})~\cite{DBLP:conf/focs/Pnueli77} is a de-facto standard for expressing temporal behaviours due to its widespread usage.
The popularity of \LTL{} stems from its desirable theoretical properties, such as efficient translation to automata and equivalence to first-order logic, as well as its interpretability, which arises from its resemblance to natural language.

Traditionally, specifications (whether temporal or not) have been manually constructed. 
This approach is error-prone, time-consuming, and requires a detailed understanding of the underlying system~\cite{BjornerH14,Rozier16}.
As a result, in recent years, there has been concentrated effort on automatically designing reliable and interpretable specifications in temporal logics.
A substantial body of research has centred on learning specifications in \LTL{}~\cite{flie,CamachoM19,scarlet} and its continuous-time extension Signal Temporal Logic (STL)~\cite{dtmethod,MohammadinejadD20}.

The primary setting of such learning frameworks is to infer specifications based on examples of trajectories generated from the underlying system.
While these frameworks are effective in learning specifications for deterministic systems, a similar approach will not be sufficient for systems with stochastic behaviour.
Specifications for these are inherently probabilistic,
for example asserting that the probability of some behaviour being observed exceeds a given threshold.
In these cases, it is not sufficient to infer a specification that characterises individual trajectories.

To accurately capture the behaviour of stochastic systems, we propose to learn temporal logic specifications from their formal models. As the logical specification formalism, we use \emph{probabilistic \LTL{}} (\PLTL{})~\cite{Var85}, which places thresholds on the probability of satisfaction of LTL formulas. The core models that we work with are discrete-time Markov chains (DTMCs).
These are commonly used for modelling in formal verification of stochastic systems.
In the context of RL, they capture the behaviour of an agent executing a specific learnt strategy (a.k.a, policy), in a probabilistic environment modelled as a Markov decision process (MDP).

We adopt the \emph{passive learning} framework~\cite{DBLP:journals/iandc/Gold78} and use a set of positive and negative DTMCs as input. 
The positive examples represent reliable probabilistic models or desirable strategies executing in stochastic environments, while the negative examples correspond to unreliable models or the behaviour of undesirable strategies. 
The learning task is to derive a concise \PLTL{} specification that captures the probabilistic temporal behaviour exhibited by the positive DTMCs while excluding the behaviour of the negative DTMCs.

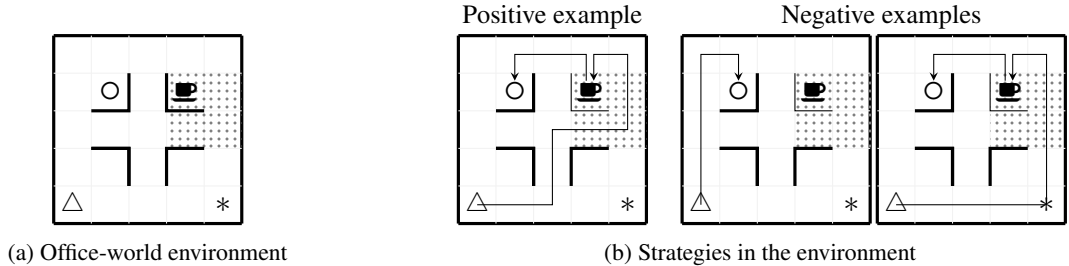
\begin{figure*}
    \centering
    \begin{subfigure}{0.25\textwidth}
        \centering
        \begin{tikzpicture}
        \def\gridsize{2.5}
        \def\gridperiod{0.5}
    
        \draw[very thick] (0, 0) -- (0, \gridsize);
        \draw[very thick] (0, 0) -- (\gridsize, 0);
        \draw[very thick] (0, \gridsize) -- (\gridsize, \gridsize);
        \draw[very thick] (\gridsize, 0) -- (\gridsize, \gridsize);
        
        \foreach \x in {0.5, 1, ..., 2.0} {
            \draw[color=gray!10] (\x, 0) -- (\x, \gridsize);
            \draw[color=gray!10] (0, \x) -- (\gridsize, \x);
        }
        \fill[pattern=dots, pattern color=gray] (1.5, 1) rectangle (2.5, 2);

        \draw[very thick] (\gridperiod, 2*\gridperiod) -- (2*\gridperiod, 2*\gridperiod) -- (2*\gridperiod, \gridperiod);
        \draw[very thick] (3*\gridperiod, \gridperiod) -- (3*\gridperiod, 2*\gridperiod) -- (4*\gridperiod, 2*\gridperiod);
        \draw[very thick] (\gridperiod, 2*\gridperiod) -- (2*\gridperiod, 2*\gridperiod) -- (2*\gridperiod, \gridperiod);
        \draw[very thick] (\gridperiod, 3*\gridperiod) -- (2*\gridperiod, 3*\gridperiod) -- (2*\gridperiod, 4*\gridperiod);
        \draw[very thick] (4*\gridperiod, 3*\gridperiod) -- (3*\gridperiod, 3*\gridperiod) --  (3*\gridperiod, 4*\gridperiod);
    
        \node at (0.5*\gridperiod, 0.5*\gridperiod) {\agent};
        \node at (3.5*\gridperiod, 3.5*\gridperiod) {\coffee};
        \node at (1.5*\gridperiod, 3.5*\gridperiod) {\office};
        \node at (4.5*\gridperiod, 0.5*\gridperiod) {\decor};
        \end{tikzpicture}
        \caption{Office-world environment}
    \end{subfigure}
    \hskip 0.5cm
    \begin{subfigure}{0.6\textwidth}
        \centering
        \begin{tikzpicture}
            \def\gridsize{2.5}
            \def\gridperiod{0.5}
        
            \draw[very thick] (0, 0) -- (0, \gridsize);
            \draw[very thick] (0, 0) -- (\gridsize, 0);
            \draw[very thick] (0, \gridsize) -- (\gridsize, \gridsize);
            \draw[very thick] (\gridsize, 0) -- (\gridsize, \gridsize);
            
        \foreach \x in {0.5, 1, ..., 2.0} {
            \draw[color=gray!10] (\x, 0) -- (\x, \gridsize);
            \draw[color=gray!10] (0, \x) -- (\gridsize, \x);
        }
        \fill[pattern=dots, pattern color=gray] (1.5, 1) rectangle (2.5, 2);
        
            \draw[very thick] (\gridperiod, 2*\gridperiod) -- (2*\gridperiod, 2*\gridperiod) -- (2*\gridperiod, \gridperiod);
            \draw[very thick] (3*\gridperiod, \gridperiod) -- (3*\gridperiod, 2*\gridperiod) -- (4*\gridperiod, 2*\gridperiod);
            \draw[very thick] (\gridperiod, 2*\gridperiod) -- (2*\gridperiod, 2*\gridperiod) -- (2*\gridperiod, \gridperiod);
            \draw[very thick] (\gridperiod, 3*\gridperiod) -- (2*\gridperiod, 3*\gridperiod) -- (2*\gridperiod, 4*\gridperiod);
            \draw (4*\gridperiod, 3*\gridperiod) -- (3*\gridperiod, 3*\gridperiod) --  (3*\gridperiod, 4*\gridperiod);
        
            \node at (0.5*\gridperiod, 0.5*\gridperiod) {\agent};
            \node at (3.5*\gridperiod, 3.5*\gridperiod) {\coffee};
            \node at (1.5*\gridperiod, 3.5*\gridperiod) {\office};
            \node at (4.5*\gridperiod, 0.5*\gridperiod) {\decor};
            
            \node at (2.5*\gridperiod, 5.5*\gridperiod) {Positive example};
            \draw[->, >=stealth, draw=black] 
            (0.5*\gridperiod, 0.5*\gridperiod) -- (2.5*\gridperiod, 0.5*\gridperiod) -- (2.5*\gridperiod, 2.5*\gridperiod) -- (4.5*\gridperiod, 2.5*\gridperiod) -- (4.5*\gridperiod, 4.5*\gridperiod) -- (3.6*\gridperiod, 4.5*\gridperiod) -- (3.6*\gridperiod, 3.8*\gridperiod);
            \draw[->, >=stealth, draw=black] 
            (3.4*\gridperiod, 3.8*\gridperiod) -- (3.4*\gridperiod, 4.5*\gridperiod) -- (1.5*\gridperiod, 4.5*\gridperiod) -- (1.5*\gridperiod, 3.8*\gridperiod);
        \end{tikzpicture}
        \hskip 0.3cm
        \begin{tikzpicture}
            \def\gridsize{2.5}
            \def\gridperiod{0.5}
        
            \draw[very thick] (0, 0) -- (0, \gridsize);
            \draw[very thick] (0, 0) -- (\gridsize, 0);
            \draw[very thick] (0, \gridsize) -- (\gridsize, \gridsize);
            \draw[very thick] (\gridsize, 0) -- (\gridsize, \gridsize);
            
        \foreach \x in {0.5, 1, ..., 2.0} {
            \draw[color=gray!10] (\x, 0) -- (\x, \gridsize);
            \draw[color=gray!10] (0, \x) -- (\gridsize, \x);
        }
        \fill[pattern=dots, pattern color=gray] (1.5, 1) rectangle (2.5, 2);
        
            \draw[very thick] (\gridperiod, 2*\gridperiod) -- (2*\gridperiod, 2*\gridperiod) -- (2*\gridperiod, \gridperiod);
            \draw[very thick] (3*\gridperiod, \gridperiod) -- (3*\gridperiod, 2*\gridperiod) -- (4*\gridperiod, 2*\gridperiod);
            \draw[very thick] (\gridperiod, 2*\gridperiod) -- (2*\gridperiod, 2*\gridperiod) -- (2*\gridperiod, \gridperiod);
            \draw[very thick] (\gridperiod, 3*\gridperiod) -- (2*\gridperiod, 3*\gridperiod) -- (2*\gridperiod, 4*\gridperiod);
            \draw (4*\gridperiod, 3*\gridperiod) -- (3*\gridperiod, 3*\gridperiod) --  (3*\gridperiod, 4*\gridperiod);
        
            \node at (5.3*\gridperiod, 5.5*\gridperiod) {Negative examples};
            \node at (0.5*\gridperiod, 0.5*\gridperiod) {\agent};
            \node at (3.5*\gridperiod, 3.5*\gridperiod) {\coffee};
            \node at (1.5*\gridperiod, 3.5*\gridperiod) {\office};
            \node at (4.5*\gridperiod, 0.5*\gridperiod) {\decor};
            
             \draw[->, >=stealth, draw=black] 
            (0.5*\gridperiod, 0.5*\gridperiod) -- (0.5*\gridperiod, 4.5*\gridperiod) -- (1.5*\gridperiod, 4.5*\gridperiod) -- (1.5*\gridperiod, 3.8*\gridperiod);
        \end{tikzpicture}
        \hspace{-1.7cm}
        \begin{tikzpicture}
            \def\gridsize{2.5}
            \def\gridperiod{0.5}
        
            \draw[very thick] (0, 0) -- (0, \gridsize);
            \draw[very thick] (0, 0) -- (\gridsize, 0);
            \draw[very thick] (0, \gridsize) -- (\gridsize, \gridsize);
            \draw[very thick] (\gridsize, 0) -- (\gridsize, \gridsize);
            
        \foreach \x in {0.5, 1, ..., 2.0} {
            \draw[color=gray!10] (\x, 0) -- (\x, \gridsize);
            \draw[color=gray!10] (0, \x) -- (\gridsize, \x);
        }
        \fill[pattern=dots, pattern color=gray] (1.5, 1) rectangle (2.5, 2);
        
            \draw[very thick] (\gridperiod, 2*\gridperiod) -- (2*\gridperiod, 2*\gridperiod) -- (2*\gridperiod, \gridperiod);
            \draw[very thick] (3*\gridperiod, \gridperiod) -- (3*\gridperiod, 2*\gridperiod) -- (4*\gridperiod, 2*\gridperiod);
            \draw[very thick] (\gridperiod, 2*\gridperiod) -- (2*\gridperiod, 2*\gridperiod) -- (2*\gridperiod, \gridperiod);
            \draw[very thick] (\gridperiod, 3*\gridperiod) -- (2*\gridperiod, 3*\gridperiod) -- (2*\gridperiod, 4*\gridperiod);
            \draw (4*\gridperiod, 3*\gridperiod) -- (3*\gridperiod, 3*\gridperiod) --  (3*\gridperiod, 4*\gridperiod);
        
            \node at (0.5*\gridperiod, 0.5*\gridperiod) {\agent};
            \node at (3.5*\gridperiod, 3.5*\gridperiod) {\coffee};
            \node at (1.5*\gridperiod, 3.5*\gridperiod) {\office};
            \node at (4.5*\gridperiod, 0.5*\gridperiod) {\decor};
            
            \draw[->, >=stealth, draw=black] 
            (0.5*\gridperiod, 0.5*\gridperiod) -- (4.5*\gridperiod, 0.5*\gridperiod) -- (4.5*\gridperiod, 4.5*\gridperiod) -- (3.6*\gridperiod, 4.5*\gridperiod) -- (3.6*\gridperiod, 3.8*\gridperiod);
            \draw[->, >=stealth, draw=black] 
            (3.4*\gridperiod, 3.8*\gridperiod) -- (3.4*\gridperiod, 4.5*\gridperiod) -- (1.5*\gridperiod, 4.5*\gridperiod) -- (1.5*\gridperiod, 3.8*\gridperiod);
        \end{tikzpicture}
        \caption{Strategies in the environment}
        \end{subfigure}

    \caption{An illustration of an office-world environment with the following features: office \office{}, coffee \coffee{}, and decoration \decor{}. The shaded part near \coffee{} is slippery and introduces stochasticity in the agent's movement.
    The positive example demonstrates a strategy where the agent \agent{} collects \coffee{} and delivers to \office{} while avoiding \decor{}, whereas the negative examples do not achieve this temporal task.}
    \label{fig:office-world}
    \end{figure*}

To illustrate our problem, we consider a simple stochastic office world environment, adapted from~\cite{DBLP:conf/ijcai/CamachoIKVM19}, shown in Figure~\ref{fig:office-world}. 
The environment consists of three features: office \office{}, coffee \coffee{}, and decoration \decor{}. 
A slippery area (shaded) near the coffee station introduces randomness in the agent's movement.
In this environment, desirable and undesirable strategies correspond to positive and negative DTMCs, respectively.
Based on this input, a possible formula could be $\lP_{\ge 0.9}[\lF(\text{\coffee} \land \lF(\text{\office}))] \land \lP_{\ge 0.9}[\lG(\neg \text{\decor})]$, which is a conjunction of two \PLTL{} formulas.
This formula uses the temporal operators $\lF$ (\emph{eventually}) and $\lG$ (\emph{always}), along with the probabilistic quantifier $\lP_{\ge 0.9}$ (\emph{at least 0.9 probability}) to state that the agent has a high probability of getting coffee and delivering it to the office while avoiding the decoration.

To address the passive learning problem for \PLTL{}, we propose a novel symbolic search algorithm comprising three key procedures. The first procedure uses grammar-based enumeration to identify candidate \LTL{} formulas. The second procedure determines threshold values for probabilistic quantifiers through probabilistic model checking, which results in a \PLTL{} formula. Finally, the third procedure constructs Boolean combinations of \PLTL{} formulas using a generalization of the set-cover problem.

To improve learning, the algorithm incorporates several heuristics for pruning the search space, including \LTL{} simplification rules, inference techniques based on model checking, and tactics for Boolean combinations. Overall, the algorithm is designed with theoretical guarantees to learn concise and interpretable \PLTL{} formulas from DTMCs.


We implement our algorithm as a tool \tool{} to leverage the grammar-based search heuristic coupled with the state-of-the-art tool PRISM~\cite{DBLP:conf/cav/KwiatkowskaNP11} for probabilistic model checking of DTMCs.
We evaluate \tool{} through two case studies. 
In the first case study, we consider learning \PLTL{} formulas to distinguish between desirable and undesirable strategies obtained using RL for a variety of temporal tasks.
In the second case study, we consider distinguishing between different variants of a probabilistic protocol.
In both case studies, \tool{} effectively infers concise and descriptive \PLTL{} formulas that explain the probabilistic temporal behaviour of the systems.
Moreover, we demonstrate how the different procedures contribute to the learning process. Additional proofs, implementation details, and experimental results are provided in Appendix~\ref{app:proofs},~\ref{app:implementation}, and~\ref{app:experimental}, respectively.

\subsection{Related Work}
There are two main areas of related work: learning temporal logics from data and explaining strategies/policies in RL.

\paragraph{Learning Temporal Logics.}
There are numerous works on learning temporal logics from data, specifically focusing on LTL~\cite{DBLP:conf/birthday/Neider025} and STL~\cite{DBLP:journals/iandc/BartocciMNN22}. Our work falls within the category of \emph{exact learning}, which seeks to infer minimal formulas that perfectly fit the data with provable guarantees.
Notable works in this category include those for 
 LTL~\cite{flie,CamachoM19,scarlet,LTL-GPU}, STL~\cite{MohammadinejadD20}, and several other temporal logics such as PSL~\cite{DBLP:conf/ijcai/0002FN20}, CTL~\cite{DBLP:conf/ijcar/PommelletSS24} and ATL~\cite{DBLP:conf/fm/BordaisNR24}.
The learning techniques primarily involve deductive methods such as constraint solving and enumerative search.

There are also several works within the category of \emph{approximate learning}, which seeks to infer formulas that fit (typically noisy) data well.
Notable works in this category include those for LTL~\cite{DBLP:conf/formats/BartocciBS14,DBLP:conf/aaai/LuoLDWPZ22,DBLP:conf/aaai/WanLDLYP24,DBLP:conf/time/Chiariello24} and STL~\cite{NenziSBB18}.
The learning techniques involve statistical optimisation, genetic algorithms, and neural network inference.

Our work considers the exact learning of probabilistic LTL, which, to our knowledge, has not been explored before.
Moreover, we introduce a new learning framework based on symbolic search guided by dedicated model checkers.

\paragraph{Explaining RL policies} 
Several approaches exist for explaining policies in reinforcement learning~\cite{DBLP:journals/csur/MilaniTVF24}. Our work falls within the category of \emph{global explanations} of pre-trained policies using formal languages. Notable works in this category include providing contrastive explanation using restricted queries in PCTL\textsuperscript{*}~\cite{DBLP:conf/ijcai/BoggessK023}, extracting finite-state machines from neural policies \cite{DBLP:conf/icml/DaneshKFK21}, and summarizing using abstract policy graphs~\cite{DBLP:conf/aaai/TopinV19}.

In contrast, our approach explains the temporal difference between policies using the full expressive power of \PLTL{}. Such explanations can also be translated to natural language~\cite{DBLP:conf/aaai/FuggittiC23}.





\section{Preliminaries}

\noindent
Let $\nat  = \{0, 1, 2, \dots\}$ be the set of natural numbers.

\subsection{Markov Chains and MDPs}


A \emph{discrete-time Markov chain} (DTMC) is a tuple $\markov=(S,s_I,P,\prop,\ell)$, where $S$ is a finite set of states, $s_I\in S$ is an initial state, $P: S\times S\mapsto [0,1]$ is a probabilistic transition function,
$\prop$ is a set of atomic propositions and $\ell: S\mapsto 2^{\prop}$ is a labelling function.
Atomic propositions will form the basis for temporal logic specifications
and the function $l$ defines the propositions that are true in each state.

A \emph{path} $\pi$ of $\markov$ is an infinite sequence of states $\pi = s_0s_1s_2\ldots\in S^\omega$ such that $P(s_i,s_{i+1})>0$ for all $i\in \nat$.
We denote the state at position $i$ of a path $\pi$ by $\pi[i] = s_i$ and
the infinite suffix starting in $\pi[i]$ as $\pi[i:] = s_{i}s_{i+1}\dots$.
The set of all paths of $\markov$ starting from state $s$ is written as $\paths^{\markov}(s)$.
In standard fashion~\cite{KSK76}, we define a probability measure
$\PM^{\markov}_s$ on the set of infinite paths $\paths^{\markov}(s)$.


A \emph{Markov decision process} (MDP) is a tuple $(S,s_I,A,P,\prop,\ell)$,
which extends a DTMC by allowing a choice between actions in each state.
The set of all actions is $A$ and the probabilistic transition function
becomes $P: S\times A\times S\mapsto [0,1]$,
where $P(s,a,s')$ is the probability to move to $s'$ when action $a$ is taken in $s$.

A \emph{strategy} (a.k.a., policy) of an MDP defines which
action is taken in each state, based on the history so far.
In their most general form, strategies are defined as functions
$\sigma: (S\times A)^* S\mapsto \Delta(A)$,
where $\Delta(A)$ is the set of distributions over $A$.
The behaviour of an MDP under a strategy $\sigma$
is defined by an \emph{induced DTMC}, which we denote by $M^\sigma$.
In general, $M^\sigma$ is infinite state.
In this paper, however, we can restrict to finite-memory strategies,
whose action choices depend only on the current state
and a finite set of memory values, since these suffice for objectives specified in LTL.
In this case, the induced DTMC $M^\sigma$ is finite~\cite{model-checking-book}.

\subsection{Probabilistic Linear Temporal Logic (\PLTL{})}

Probabilistic Linear Temporal Logic (\PLTL{}) is the probabilistic variant of the popular logic \LTL{}, and is commonly used to express the temporal behaviour of probabilistic systems.
A PLTL formula takes the form $\lP_{\bowtie p}[\varphi]$,
stating that the probability with which LTL formula $\varphi$ is satisfied
meets the probability threshold $\bowtie p$.
For example, PLTL formula $\lP_{\geq 0.9}[\lF(a\land\lF b)]$
means that the probability of observing proposition $a$ and then $b$ is at least 0.9.

In this work, we learn specifications expressed in an extension of \PLTL{},
which we call \PLTLB{}, that allows positive Boolean combinations of \PLTL{} formulas.


Formally, the syntax and semantics of these logics are defined as follows.
Firstly, LTL formulas $\varphi$ are defined inductively using the following grammar:
\begin{align*}
\varphi &::= p ~|~ \neg \varphi ~|~ \varphi \lor \varphi ~|~ \varphi \land \varphi ~|~ \lX\varphi ~|~ \varphi\lU\varphi,
\end{align*}
where $p\in\prop$ is an proposition, $\neg$ (not), $\lor$ (or) and $\land$ (and) are standard Boolean operators, and $\lX$ (neXt), $\lU$ (Until) are standard temporal operators.
We allow the standard temporal operators $\lF$ (Finally) and $\lG$ (Globally) as syntactic sugar, where $\lF\varphi:=\ltrue\lU\varphi$ and $\lG\varphi:=\neg\lF\neg\varphi$.

We interpret LTL formulas over paths of a DTMC.
The satisfaction of LTL formula $\varphi$ by (infinite) path $\pi$
is defined inductively as follows:
\begin{align*}
\pi\models p &\text{ iff } p\in\ell(\pi[0]) \\
\pi\models \neg \varphi &\text{ iff } \pi\not\models \varphi \\
\pi\models \lX \varphi &\text{ iff } \pi[1:]\models \varphi \\
\pi\models \varphi\lU \varphi' &\text{ iff } \text{there exists } i\in\nat: \pi[i:]\models \varphi' \\&\hspace{1cm}\text{ and for all } j<i: \pi[j:]\models\varphi
\end{align*}
We interpret Boolean combinations in the standard fashion and therefore omit the definitions.

A \PLTLB{} formula $\Phi$ is defined as:
\[\Phi ::= \lP_{\bowtie p}[\varphi] ~|~ \Phi \lor \Phi ~|~ \Phi \land \Phi,\] 
where $\bowtie \,\, \in \{<,>,\leq,\geq\}$, $p\in[0,1]$ is a probability threshold and $\varphi$ is an \LTL{} formula.

A \PLTLB{} (or \PLTL{) formula is interpreted over the states of a DTMC.
The satisfaction of \PLTLB{} formula $\Phi$ by a state~$s$
is defined as follows: 
\begin{align*}
s\models \lP_{\bowtie p}[\varphi] &\text{ iff } \PM^{\markov}(s \models \varphi) \bowtie p, 
\end{align*}
where $\PM^{\markov}(s \models \varphi) = \PM^{\markov}_s(\{\pi\in \paths^{\markov}(s)~|~\pi\models\varphi\})$
denotes the probability that LTL formula $\varphi$ is satisfied by a path starting in state $s$ of $\markov$.
We say that a DTMC $\markov$ satisfies a \PLTLB{} formula $\Phi$ if, for the initial state $s_I$ of $\markov$, $s_I\models\Phi$.

\section{Passive Learning of \PLTLB{} Formulas}

We frame the problem of learning a \PLTLB{} formula as a typical \emph{passive learning} problem~\cite{DBLP:journals/iandc/Gold78}.
Apart from being a fundamental learning problem, passive learning forms a key subroutine in other learning frameworks, such as active learning~\cite{CamachoM19} and learning from positive examples~\cite{ltl-from-positive-only}.

In this problem, we rely on a \emph{sample} $\sample = (P,N)$ consisting of a set $P$ of positive DTMCs and a set $N$ of negative DTMCs.
We define sample size $|\sample|$ as the total number of DTMCs in $\sample$.
The goal is to learn a concise \PLTLB{} formula~$\Phi$ that is \emph{consistent} with $\sample = (P, N)$, i.e., for all $\markov\in P$, $\markov\models \Phi$, and for all $\markov\in N$, $\markov\not\models\Phi$. 

To quantify conciseness, we measure the \emph{size} $\size{\Phi}$ of \PLTLB{} formulas $\Phi$.
To avoid checking redundant formulas, our learning algorithm uses LTL in \emph{negation normal form} (NNF), a standard syntactic form where negation applies only to atomic propositions.
We define the size of an \LTL{} formula by the number of operators $\circ\in\{\lF,\lX,\lG,\lU,\land,\lor\}$ and literals $\Lambda=\{p,\neg p~|~p\in \prop\}$ in the formula. For instance, the formula $\lF(p\land \lF (\neg q))$ has size 5.
The size of a \PLTLB{} formula $\Phi$ is defined exactly the same way.

We now formally define the problem of passive learning of \PLTL{} formulas.
\begin{problem}\label{prob:pass-learning}
Given a sample $\sample = (P,N)$, size bound $K$ and propositions $\prop$, learn a minimal \PLTLB{} formula $\Phi$ over $\prop$ such that: (i) $\Phi$ is consistent with $\sample$, and (ii) $|\Phi|\leq K$.
\end{problem}
A solution to Problem~\ref{prob:pass-learning} (which is not necessarily unique) is a concise \PLTLB{} formula $\Phi$ that distinguishes between the probabilistic temporal behaviour of the positive and negative DTMCs.
The size bound $K$ ensures three attributes for the solution formula $\Phi$: (i) it does not get too large, (ii) it does not overfit to the sample, and (iii) it makes the passive learning problem decidable, ensuring a terminating algorithm.



\section{The Learning Algorithm}

We now describe the learning algorithm that we propose to solve Problem~\ref{prob:pass-learning}.
Figure \ref{fig:overview} provides a high-level overview of our algorithm.
The algorithm consists of three main procedures: (i) \emph{grammar-based enumeration}, which efficiently enumerates through the space of LTL formulas, (ii) \emph{probabilistic threshold search}, which employs probabilistic model checking to determine whether a formula is consistent with the given sample, and (iii) \emph{Boolean set cover}, which constructs Boolean combinations of PLTL formulas to form a consistent formula. 
The algorithm iterates over formulas of increasing size, starting from 1, using GBE, and then checks the consistency of the formulas using PTS and BSC.
We now describe each of these procedures in detail.

\begin{figure}
    \centering
    \begin{tikzpicture}[node distance=2.7cm]
        \tikzstyle{block} = [rectangle, draw, text centered, rounded corners, font=\small]
        \tikzstyle{line} = [draw, -latex']
        \node[block, text width=5.7em] (block1) {Grammar-based enumeration (GBE)};
        \node [block, text width=5em, right of=block1] (block2) {Probabilistic threshold search (PTS)};
        \node [block, text width=5em, right of=block2] (block3) {Boolean set cover (BSC)};

        \draw[->] ++(-1.9,0) to node[midway, above, font=\scriptsize] {$n=1$} (block1.west);
        \draw[->] (block1.north east) to [out=60, in=120] node[midway, above, font=\scriptsize] {$\formulalist_n$} (block2.north west);
        \draw[->] (block2.south west) to [out=-120, in=-60] node[midway, below, font=\scriptsize] {$\discardlist_n$} (block1.south east);
        \draw[->] (block2) to node[midway, above, font=\scriptsize] {$\boollist_n$} (block3);
        \draw[->] (block3.north) -- ++(0,1) node at ($(block3.north) + (-3.5,1.2)$) [right, font=\scriptsize] {$n=n+1$} -| (block1.north);
        \draw[->] (block2.south) -- ++(0,-0.5) node[midway, right, font=\scriptsize] {\PLTL{} $\Phi^*$};
        \draw[->] (block3.south) -- ++(0,-0.5) node[midway, right, font=\scriptsize] {\PLTLB{} $\Phi^*$};
    \end{tikzpicture}
    \caption{The high-level overview of the learning algorithm. The set $\formulalist_n$ consists of formulas of size $n$, the set $\discardlist_n$ consists of discarded formulas, and the set $\boollist_n$ consists of formulas for Boolean combinations. The procedures PTS and BSC output a consistent \PLTL{} and \PLTLB{} formula $\Phi^*$, respectively, if they find one.}
    \label{fig:overview}
\end{figure}
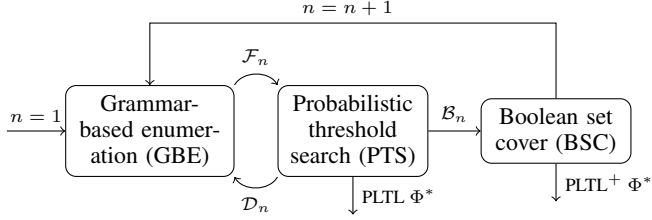

\subsection{Grammar-based Enumeration for \LTL{}}

The grammar-based enumeration (GBE) procedure incrementally explores the space of \LTL{} formulas. Since the number of syntactically distinct formulas grows exponentially with size\footnote{Asymptotically $\frac{7^n\sqrt{14}}{2\sqrt{\pi n^3}}$~\cite{DBLP:books/daglib/0023751}}, GBE employs pruning techniques to manage the search space efficiently.

Importantly, GBE relies on the \emph{nesting depth} (or depth, for brevity) of temporal operators in formulas.
We define the nesting depth $\depth{\varphi}$ for \LTL{} recursively as: $\depth{l} = 0$ for $l\in\literals$, 
$\depth{\varphi \op \varphi'} = \max(\depth{\varphi}, \depth{\varphi'})$ for $\op \in \{\land,\lor\}$, 
$\depth{\op\varphi} = 1 + \depth{\varphi}$ for $\op \in \{\lX, \lF, \lG\}$, $\depth{\varphi\lU\varphi'} = 1 + \max(\depth{\varphi}, \depth{\varphi'})$. 
For instance, the formula $\lF(p\land \lF (\neg q))$ has depth 2.

The nesting depth is essential for managing the search space effectively as well as ensuring the practical applicability of the formulas. 
Heavily nested formulas (e.g., those with a depth $>3$) are considered hard to interpret~\cite{CamachoM19} and are also uncommon in widely used LTL patterns~\cite{DBLP:conf/fmsp/DwyerAC98}. 
Therefore, GBE incorporates maximum depth $D$ as a parameter.

To build formulas $\formulalist_n$ of size $n$ and all depths $d\leq D$, GBE employs a bottom-up dynamic programming approach.
Specifically, the set $\formulalist_n$ is built as a union of subsets $\formulalist_n^d$ of LTL formulas of size $n$ and depth $d$.
GBE initializes $\formulalist_1^0 := \literals$ and $\formulalist_1^d := \emptyset$ for $0<d\leq D$.
It then inductively combines formulas from $\formulalist_n^d$ of different depths to form larger formulas.

\begin{algorithm}[tb]
    \caption{Inductive step in GBE}
    \label{alg:grammar-based-search}
    \textbf{Input}: $\formulalist_n$, Max depth $D$\\
    \vspace{-0.4cm}
    \begin{algorithmic}[1] 
        \FOR{$d = 0$ to $D$}
            \STATE $\formulalist_{n+1}^d = \emptyset$
                \FOR{$\varphi\in\formulalist_n^{d-1}$}
                    \STATE Construct $\psi = \circ\varphi$ for $\circ\in\{\lX, \lF, \lG\}$
                    \STATE Add $\psi$ to $\formulalist_{n+1}^d$ if \emph{temporal simplify} does not hold
                \ENDFOR
            \FOR{$k = 1$ to $n-1$}
            \FOR{$\varphi\in\formulalist_k^{d-1}$ and $\varphi'\in\bigcup_{d'<d}\formulalist_{n-k}^{d'}$ }
                \STATE $\psi = \varphi\lU\varphi'$
                \STATE Add $\psi$ to $\formulalist_{n+1}^{d}$ if \emph{Boolean simplify} does not hold
            \ENDFOR
            \FOR{$\varphi\in\formulalist_k^{d}$ and  $\varphi'\in\bigcup_{d'\leq d}\formulalist_{n-k}^{d'}$ }
                \STATE $\psi = \varphi\circ\varphi'$ for $\circ\in\{\land,\lor\}$
                \STATE Add $\psi$ to $\formulalist_{n+1}^{d}$ if \emph{Boolean simplify} does not hold
            \ENDFOR
            \ENDFOR
        \ENDFOR
        \STATE \textbf{return} $\formulalist_{n+1}$
    \end{algorithmic}
\end{algorithm}

The inductive step of GBE $\formulalist_{n+1}$ is outlined in Algorithm~\ref{alg:grammar-based-search}. It follows the LTL grammar, adding operators to smaller formulas to construct larger ones. The algorithm uses two heuristics to eliminate semantically equivalent ($\equiv$) formulas, where $\varphi \equiv \varphi'$ if and only if $\pi \models \varphi \leftrightarrow \pi \models \varphi'$ for any path $\pi\in (2^{\prop})^{\omega}$. These heuristics are briefly described here.

The first heuristic, \emph{temporal simplification}, removes redundant formulas by applying syntactic rules to rewrite LTL formulas into normal forms~\cite[Fig. 5.7]{model-checking-book}. For example, $\lF\lF(p) \equiv \lF(p)$, $\lF\lX(p) \equiv \lX\lF(p)$, and $\lF\lG\lF(p) \equiv \lG\lF(p)$ (see~\cite[Sec. 5.4]{spotrewriterules} for the full list). If a constructed formula is not in simplified form, it is discarded.
These checks are constant-time operations, making them highly efficient.

The second heuristic, \emph{Boolean simplification}, removes redundant Boolean combinations such as $\varphi_1 \land \varphi_2$ or $\varphi_1 \lor \varphi_2$, where $\varphi_1 \equiv \varphi_2$ or $\varphi_1 \equiv \neg \varphi_2$. It checks the syntactic equality of $\varphi_1$ and $\varphi_2$ via a linear-time scan of their syntax trees. 
It then checks semantic equivalence, which can be doubly exponential in formula size but is efficiently handled by modern \LTL{} satisfiability checkers~\cite{duret.22.cav}.

We formalise the completeness of the GBE procedure below\footnote{This result establishes GBE's complete search in isolation; when combined with PTS, as we see later in Section~\ref{subsec:PTS}, several formulas are pruned using heuristics.}, which can be proved by induction on formula size.
\begin{lemma}\label{lem:gbe-completeness}
$\bigcup_{n\leq N'}\formulalist_n$ computed by GBE consists of all semantically distinct formulas of size $\leq N'$ and depth $\leq D$.
\end{lemma}

\subsection{Probabilistic Threshold Search for~\PLTL{}}
\label{subsec:PTS}
The probabilistic threshold search (PTS) procedure steps through the formulas generated by GBE and evaluates the likelihood of the formulas being satisfied/consistent. The main steps of PTS are outlined in Algorithm~\ref{alg:prob-search}.

PTS first computes the probability measure for a formula $\varphi\in\formulalist_n$  
for each $\markov$ in the given sample.
More specifically, PTS computes the vector $\probvector^{\markov,\varphi}: S\to [0,1]$, mapping each state $s \in S$ of $\markov$ to $\PM^{\markov}(s \models \varphi)$. We use $v^{\markov,\varphi}_I = \probvector^{\markov,\varphi}(s_I)$ to denote the probability for the initial state $s_I$ of $M$. 

Our implementation, which is based on the PRISM tool~\cite{DBLP:conf/cav/KwiatkowskaNP11},
deploys standard probabilistic LTL model checking procedures~\cite{model-checking-book}.
First, the LTL formula $\varphi$ is translated into an equivalent deterministic Rabin automaton (DRA) $\automata_\varphi$.
Then the product DTMC $\markov\times\automata_\varphi$, which combines the DTMC and DRA, is constructed and solved using standard numerical methods based on value iteration.



PTS exploits the computed vector $\probvector^{\markov,\varphi}$ to search for a formula that has a higher probability of satisfaction in the positive examples than in the negative examples.
For this, it computes the minimum probability $p_{\varphi} = \min\{v^{\markov,\varphi}_I~|~ \markov\in P\}$
of satisfaction of $\varphi$ among the samples in the set $P$, and
the maximum probability $n_{\varphi} = \max\{v^{\markov,\varphi}_I~|~ \markov\in N\}$ of satisfaction of $\varphi$ among the samples in the set $N$.

To identify a significant probabilistic difference in the temporal behaviour, PTS employs a small tolerance parameter $\delta\in (0,0.1)$.
We can understand this parameter using the introductory example from Figure~\ref{fig:office-world}.
In this example, the probability of reaching the office, i.e., satisfying $\lF(\text{\office})$, can be slightly lower in the positive example (say 0.94) compared to the negative examples (say, 1.0 and 0.95).
This small difference could arise because the agent takes a slightly longer slippery route in the positive example than in the second negative example.
However, $\lF(\text{\office})$ is not a primary distinguishing factor between the positive and negative examples and must not be considered by PTS.

Thus, PTS checks if the difference $p_{\varphi} - n_{\varphi}$ between the probability of satisfaction of $\varphi$ in the positive and negative examples is greater than the tolerance $\delta$.
If indeed $p_{\varphi} - n_{\varphi}>\delta$, then PTS outputs the formula $\Phi = \lP_{> m_{\varphi}}[\varphi]$, where the threshold $m_{\varphi} = \frac{p_{\varphi}+n_{\varphi}}{2}$.
While any threshold between $p_{\varphi}$ and $n_{\varphi}$ could be chosen, the choice of the mean of the two values is to reduce overfitting to the input sample.

We state the soundness of the PTS procedure as follows.
\begin{lemma}\label{lem:pts-correctness}
    If PTS returns a \PLTL{} formula $\Phi$, then $\Phi$ is consistent with sample $\sample$.
\end{lemma}
\begin{proof}
    PTS always returns a formula of the form $\Phi=\lP_{> m_{\varphi}}[\varphi]$, where $m_{\varphi} = \frac{p_{\varphi}+n_{\varphi}}{2}$  and $p_{\varphi} - n_{\varphi}>\delta$.
    We can state that 
    \begin{align*}
    \forall M\in P,v^{\markov,\varphi}_I > m_\varphi \text{ iff } \PM^{\markov}(s_I\models\varphi) > m_\varphi \text{ iff } M\models \Phi,\\
    \forall M\in N, v^{\markov,\varphi}_I < m_\varphi \text{ iff } \PM^{\markov}(s_I\models\varphi) < m_\varphi \text{ iff } M\not\models \Phi.
    \end{align*}
\end{proof}
Note that PTS restricts the search to only \PLTL{} formulas of the form $\lP_{> p}[\varphi]$.
The relation $\geq$ is not required due to the non-zero parameter $\delta$, while the $<$ relation can be derived from the $>$ relation and the dual \LTL{} formula $\neg \varphi$ using the relation $\lP_{< p}[\varphi] \equiv \lP_{> 1-p}[\neg \varphi]$.

If a formula $\varphi$ is not consistent, PTS discards it by adding it to $\discardlist_n$ if it is not useful for the next GBE iterations; otherwise, it adds $\varphi$ to $\boollist_n$ for Boolean combinations.
We briefly discuss the heuristics used for discarding formulas.

This heuristic, \emph{inconsistency removal}, discards $\varphi$ if the following condition holds: $\probvector^{\markov,\varphi}\equiv \mathbf{0}$ for each $\markov\in P$, or $\probvector^{\markov,\varphi}\equiv \mathbf{1}$ for each $\markov\in N$, where $\mathbf{0}$ and $\mathbf{1}$ are the vectors with all zeros and all ones, respectively.
In simpler terms, $\varphi$ is discarded if it is unsatisfiable in any state of the positive DTMCs or universally satisfied in all states of the negative DTMCs. Such formulas cannot be meaningfully combined in subsequent iterations, as stated below for the positive cases; a similar argument applies to the negative cases.
\begin{lemma}\label{lem:inconsistency-removal}
    Let $\probvector^{\markov,\varphi}\equiv \mathbf{0}$ for each $\markov\in P$. Then $\varphi$ cannot be a subformula of a minimal consistent \PLTL{} formula.
    \end{lemma}

Examples of formulas that can be discarded from the introductory example include $\lF(\text{\decor} \land \text{\coffee})$, $\lG(\text{\coffee})$, and $\lG(\text{\decor})$ since they never hold in any state of the positive DTMC.


\begin{algorithm}[tb]
    \caption{Probabilistic Threshold Search (PTS)}
    \label{alg:prob-search}
    \textbf{Input}: $\formulalist_n$, Probabilistic tolerance $\delta$\\
    \vspace{-0.4cm}
    \begin{algorithmic}[1]
        \FOR{$d = 0$ to $D$, $\varphi\in\formulalist_n^d$}
                \STATE Compute $\probvector^{\markov,\varphi}$ for each $\markov$ in $\sample$
                \STATE $p_{\varphi} = \underset{\markov\in P}{\min}\{v^{\markov,\varphi}_I \}$, $n_{\varphi} = \underset{\markov\in N}{\max}\{v^{\markov,\varphi}_I\}$
                \IF{$p_{\varphi} - n_{\varphi}> \delta$}
                    \STATE \textbf{return} $\Phi = \lP_{> \frac{p_{\varphi}+n_{\varphi}}{2}}[\varphi]$
                \ELSE
                    \STATE Add $\varphi$ to $\discardlist_n$ if \emph{inconsistency removal} holds
                    \STATE Add $\varphi$ to $\boollist_n$ otherwise
                \ENDIF
        \ENDFOR
    \end{algorithmic}
\end{algorithm}

\subsection{Boolean Set Cover for \PLTLB{}}
The Boolean Set Cover (BSC) procedure combines \PLTL{} formulas using Boolean operations. 
We adapt this procedure, originally introduced in~\cite{scarlet}, to accommodate for probability thresholds.
The steps of our algorithm are detailed in Algorithm~\ref{alg:boolean-set-cover}.

First, BSC discards formulas from $\boollist_n$ that are not useful for Boolean combinations.
For this, it uses a condition similar to, but weaker than, inconsistency removal used in PTS:
$v^{\markov,\varphi}_I = 0$ for all $\markov \in P$, or $v^{\markov,\varphi}_I = 1$ for all $\markov \in N$. 

For the remaining formulas, BSC assesses how close they are to being a consistent formula. 
To do this, it relies on the function $c(\varphi, r)$, which quantifies the quality of $\varphi$ with probability threshold $r$, defined as follows:
\[
c(\varphi, r) = \Big[\sum\limits_{M\in P} \small{\llbracket v^{\markov,\varphi}_I > r\rrbracket} + \sum\limits_{M\in N} \small{\llbracket v^{\markov,\varphi}_I}<r\rrbracket\Big],
\]
where $\llbracket\cdot\rrbracket$ denotes the Iverson bracket, evaluating to 1 if the condition holds, and 0 otherwise.
We have $c(\varphi,r)=|\sample|$ if and only if $\lP_{>r}[\varphi]$ is consistent with $\sample$.

BSC computes, for each \LTL{} formula $\varphi\in \boollist_n$, a maximal probability threshold $r^* = \arg\max_{r\in(0,1)} c(\varphi,r)$ that maximizes consistency with $\sample$. This can be computed via a linear scan over the sorted list of probabilities $v^{\markov,\varphi}_I$ for $M$ in $\sample$.

BSC then constructs the \PLTL{} formula $\Phi=\lP_{> r^*}[\varphi]$ along with its score $\score(\Phi) = c(\varphi,r^*)/(1+\sqrt{|\Phi|})$ and adds it to a heap $\heap$.
The scoring function and the subsequent steps of BSC are as in~\cite{scarlet}.
Briefly, a maximum Boolean combination limit $L$ is considered. 
The \PLTL{} formulas with the $L$ highest scores are selected, and combined as disjunctions and conjunctions with all formulas in $\heap$.

We have the soundness of BSC for \PLTL{} based on~\cite{scarlet}.
\begin{lemma}\label{lem:bsc-correctness}
If BSC returns a \PLTLB{} formula $\Phi$, then $\Phi$ is consistent with sample $\sample$.
\end{lemma}

\begin{algorithm}[tb]
    \caption{Boolean Set Cover (BSC) for \PLTL{}}
    \label{alg:boolean-set-cover}
    \textbf{Input}: $\boollist_n$, Max size $K$, Max Limit $L$\\
    \vspace{-0.4cm}
    \begin{algorithmic}[1]
        \STATE Discard formulas from $\boollist_n$ not suitable for bool comb
        \FOR{$\varphi\in \boollist_n$}
            \STATE Compute $\Phi = \lP_{> r^*}[\varphi]$, score $\score(\Phi)$ and add to $\heap$
        \ENDFOR
        \STATE $\heap^\ast\gets$ Highest $L$ formulas in $\heap$ w.r.t score $\score$
		\FOR{$\Psi \in \heap$ and $\Phi \in \heap^\ast$}
            \STATE $\Phi':=\Psi \circ \Phi$ for $\circ\in\{\land,\lor\}$
            \IF{$|\Phi'|\leq K$ and $\Phi'$ is consistent}
            \STATE Store $\Phi'$ as consistent and update $K\gets|\Phi'|-1$
            \ENDIF
		\ENDFOR
    \end{algorithmic}
\end{algorithm}

\paragraph{Theoretical guarantees.} We state the guarantees of our algorithm with respect to the search space $\Theta(K,D,\delta)$ of \PLTL{} formulas constrained by the considered parameters, i.e., size $\leq K$, depth $\leq D$ and tolerance $>\delta$.
\begin{theorem}\label{thm:guarantees}
    Given sample $\sample$, size $K$, depth $D$, and tolerance $\delta$, our learning algorithm has the following guarantees:
    \begin{itemize}
        \item (soundness) if it returns a \PLTLB{} formula $\Phi$, then $\Phi$ is consistent with $\sample$ and $|\Phi|\leq K$, and
        \item (completeness and minimality) if there exists a \PLTL{} formula in $\Theta(K,D,\delta)$ consistent with $\sample$, then it returns a minimal \PLTLB{} formula.
    \end{itemize}
\end{theorem}
\begin{proof}
Soundness follows from the correctness of PTS (Lemma~\ref{lem:pts-correctness}) and BSC (Lemma~\ref{lem:bsc-correctness}) in outputting a consistent formula.
Completeness is ensured by the exhaustive enumeration by GBE (Lemma~\ref{lem:gbe-completeness}), discarding only inconsistent formulas (Lemma~\ref{lem:inconsistency-removal}).
Minimality follows from the complete iterative search over increasing formula sizes (see Fig.~\ref{fig:overview}).
\end{proof}


\section{Evaluation}
\label{sec:implementation}

In this section, we evaluate the capability of our learning algorithm to infer concise \PLTLB{} formulas from samples of DTMCs. To this end, we developed a prototype tool, \tool{}\footnote{\url{https://github.com/rajarshi008/PriTL}}, implemented in Python3, which integrates the three procedures, GBE, PTS and BSC, of the learning algorithm.
For heuristics in GBE, we rely on \LTL{} simplification and satisfaction features from the SPOT library~\cite{duret.22.cav}.
For \LTL{} model checking of DTMCs in PTS, we rely on the PRISM tool~\cite{DBLP:conf/cav/KwiatkowskaNP11}, using its (default) hybrid model checking engine.

To the best of our knowledge, no existing tool can directly learn arbitrary temporal specifications from DTMCs. To evaluate \tool{}'s ability to learn concise and distinguishing formulas, we tested it on strategies generated within a stochastic environment and on various variants of a probabilistic model.
For all experiments, we set the maximum depth $D=2$, tolerance $\delta=0.05$, and a Boolean combination limit $L=10$. 
If \tool{} identifies multiple minimal formulas, it returns the one with the highest probability difference, $p_\varphi-n_\varphi$.
In case of a tie, \tool{} returns all valid formulas.

We conducted all experiments on a MacBook Pro M3 (macOS 14.6.1) with 18~GB RAM. We provide additional implementation details in Appendix~\ref{app:implementation}.

\subsubsection{Learning from strategies in stochastic environment}
For this experiment, we focus on strategy DTMCs generated via non-Markovian reinforcement learning algorithms. 
Specifically, we utilize different Q-learning algorithms proposed by~\cite{DBLP:conf/ijcai/ShaoK23} that are capable of generating optimal strategies for LTL tasks.
We present details of the strategy training process in Appendix~\ref{app:strategy-exp}.

As the underlying MDP, we select the widely used OpenAI Gym frozen lake environment~\cite{DBLP:journals/corr/BrockmanCPSSTZ16}, employing the same layout as in~\cite{DBLP:conf/ijcai/ShaoK23}. 
This environment is an 8$\times$8 gridworld,  where an agent navigates a slippery frozen lake, introducing stochasticity: with a 1/3 probability, the agent moves in the intended direction, and with a 1/3 probability, it deviates sideways.
The environment includes three key features: two campsites, $a$ and $b$, and several holes $h$.

\begin{table*}[h]
    \centering
    \caption{Summary of the learning from strategy DTMCs for correct and incorrect tasks on Frozen Lake.}
    \label{tab:correct-incorrect}
    \setlength{\tabcolsep}{3pt}
    \small
    \begin{tabular}{cccccc}
    \toprule
    Correct task for $P$ & Incorrect task(s) for $N$& Learned \PLTLB{} Formula & State space of $\sample$ & \LTL{} Searched & Time (sec)\\
    \midrule
    $\lF(a) \land \lG(\neg h)$ & $\lF(a)$ & $\lP_{>0.76}[\lG(\neg h)]$ & $2.5\cdot 10^3$ & $24/24$ & $1.77$  \\
    $\lF(a) \land \lG(\neg h)$ & $\lF(a)$, $\lG(\neg h)$  & $\lP_{>0.76}[\neg h \lU a]$ & $1.9\cdot 10^3$ & $110/186$ & $3.39$ \\
    $\lF\lG(a) \land \lG(\neg h)$ & $\lF(a) \land \lG(\neg h)$ & $\lP_{>0.49}[\lF\lG(a)]$, $\lP_{>0.49}[\lG\lF(a)]$ & $2.7\cdot 10^3$ & $112/186$ & $4.93$ \\
    $\lG\lF(a) \land \lG\lF(\neg a) \land \lG(\neg h)$ & $\lF\lG(a)\land \lG(\neg h)$ & $\lP_{>0.5}[\lG\lF(\neg a)]$ & $0.9\cdot 10^3$ & $50/186$ & $3.9$ \\
    $\lF(a) \land \lF(b) \land \lG(\neg h)$ & $\lF(a) \land \lG(\neg h)$, $\lF(b) \land \lG(\neg h)$ & $\lP_{>0.5}[\lF(a)]\land \lP_{>0.99}[\lF(b)]$ & $3.2\cdot 10^3$ & $476/7314$ & $24.1$ \\
    \bottomrule
    \end{tabular}
\end{table*}

We evaluate \tool{} on two distinct applications: (i) learning from strategies trained on correct and incorrect LTL tasks, and (ii) learning from optimal and suboptimal strategies for the same LTL task.
For both (i) and (ii), we set the propositions $\prop = \{a,b,h\}$ and formula size bound $K=10$.

For 
application (i), we identify several desirable LTL tasks and designate them as correct tasks. 
As incorrect tasks, we select LTL tasks that are less precise than their correct counterparts. 
For example, the correct task $\lF(a) \land \lG(\neg h)$ requires reaching campsite $a$ while always avoiding holes $h$, whereas the incorrect task $\lF(a)$ specifies a weaker condition of reaching the campsite, which may result in falling into holes.
The first two columns of Table~\ref{tab:correct-incorrect} list the considered correct and incorrect tasks, respectively.
For each correct and incorrect task, we generate 10 positive and 10 negative optimal strategy DTMCs, respectively, using the CF+KC Q-learning algorithm~\cite{DBLP:conf/ijcai/ShaoK23}, known for its fast convergence to optimality. 
The cumulative state space of the samples is of the order of $10^3$ (see fourth column of Table~\ref{tab:correct-incorrect}).

We present the learned \PLTLB{} formulas for each task in Table~\ref{tab:correct-incorrect}. 
For the first two tasks, \tool{} inferred \PLTLB{} formulas with safety properties, $\lG(\neg h)$ and $\neg h \lU a$, which were violated in the negative examples.
In subsequent tasks, \tool{} inferred formulas that indicate the specific requirements missing in the negative examples.
These include repeated reachability $\lG\lF(a)$ instead of simple reachability $\lF(a)$, performing two tasks simultaneously $\lF(a) \land \lF(b)$ instead of just one $\lF(a)$ or $\lF(b)$, etc.
Overall, \tool{} successfully produced concise formulas that explain the differences between strategies trained on different tasks.

For 
application (ii), we identify some more desirable LTL tasks (in Figure~\ref{fig:time_results}) and generate optimal and suboptimal strategies for each.
We extract strategy DTMCs from intermediate episodes of the KC Q-learning algorithm since it has relatively slower convergence to optimality~\cite{DBLP:conf/ijcai/ShaoK23}, thereby often yielding sub-optimal strategies.
We considered a strategy that achieves a high probability (i.e., $p\geq0.95$) as optimal, while one that achieves a lower probability (i.e., $0.5\leq p \leq0.9$) as suboptimal.
Overall, for each LTL task, we collected at least 30 positive DTMCs and 30 negative DTMCs corresponding to optimal and suboptimal strategies, respectively.
We evaluated our algorithm using varying sample sizes $|\sample|$, ranging from 10 to 60 DTMCs per sample, with an equal split between positive and negative examples.

We present the runtime for varying sample sizes in Figure~\ref{fig:time_results}.
For the tasks $\neg h \lU a$ and $\lG\lF(a) \lor \lG\lF(b)$, \tool{} consistently inferred the formulas $\lP_{>0.9}[\lG(\neg h)]$ and $\lP_{>0.93}[\lG\lF(a) \lor \lG\lF(b)]$, respectively, across all samples.
The runtime for these increased linearly with the sample size.

For the task $\lF(a \land \lF(b))$, \tool{} inferred $\lP_{>0.97}[\lF(b)]$ for a sample size of 10 and a more precise formula, $\lP_{>0.94}[\lF(a)] \land \lP_{>0.97}[\lF(b)]$, for larger sizes. 
Similarly, for the task $\lG\lF(a) \land \lF(b)$, \tool{} inferred $\lP_{>0.92}[\lF(b)]$, for smaller samples ($\leq40$), whereas it inferred a more precise formula, $\lP_{>0.92}[\lG\lF(a)] \land \lP_{>0.95}[\lF(b)]$ for larger samples (50 and 60).
Both tasks showed a runtime spike due to the change in the inferred formula, deviating from the linear trend in other tasks.
Moreover, in both cases, the more precise formula was inferred by the BSC procedure.
Overall, \tool{} successfully inferred expected \PLTLB{} formulas, with runtime generally scaling linearly with sample size.
\begin{figure}
    \centering
    \begin{tikzpicture}
        \begin{semilogyaxis}[
            xlabel={Number of DTMCs},
            ylabel={Time (sec)},
            legend style={at={(1.05,0.4)}, anchor=west},
            ymajorgrids=true,
            grid style=dashed,
            width=0.27\textwidth,
            height=0.27\textwidth,
            xtick={10, 20, 30, 40, 50, 60},
            tick style={major tick length=2pt, minor tick length=1pt}, 
            tick label style={font=\small} 
        ]
            \addplot[color=red, mark=*, mark options={opacity=0.5}] coordinates {
                (10,0.98)
                (20,1.8)
                (30,2.7)
                (40,4.06)
                (50,4.66)
                (60,8.95)
            };
            \addlegendentry{$\neg h \lU a$}

            \addplot[color=blue, mark=*, mark options={opacity=0.5}] coordinates {
                (10,0.99)
                (20,13.84)
                (30,21.04)
                (40,28.89)
                (50,35.93)
                (60,44.05)
            };
            \addlegendentry{$\lF(a \land \lF(b))$}

            \addplot[color=green, mark=*, mark options={opacity=0.5}] coordinates {
                (10,0.95)
                (20,1.91)
                (30,2.82)
                (40,4.26)
                (50,206.03)
                (60,246.87)
            };
            \addlegendentry{$\lG\lF(a) \land \lF(b)$}

            \addplot[color=orange, mark=*, mark options={opacity=0.5}] coordinates {
                (10,27.7)
                (20,53.46)
                (30,87.21)
                (40,125.66)
                (50,157.01)
                (60,188.62)
            };
            \addlegendentry{$\lG\lF(a) \lor \lG\lF(b)$}    

        \end{semilogyaxis}
        \node[anchor=south west] at (rel axis cs:1.42,0.82) {LTL Tasks};
    \end{tikzpicture}
    \caption{Runtime comparison for strategies generated from varying formulas and varying sample sizes.}
    \label{fig:time_results}
\end{figure}
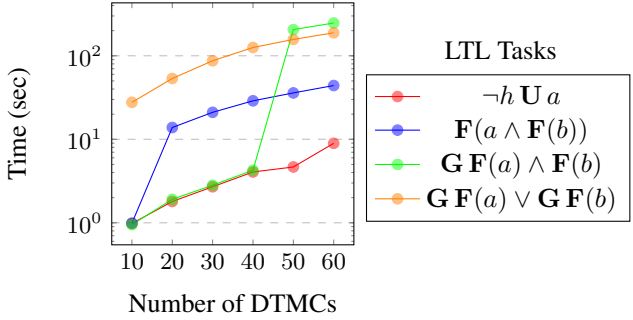


We also briefly discuss how different parts of \tool{} contribute to the learning process. 
PTS dominates the running time, while GBE and BSC take negligible time. 
For example, inferring $\lP_{>0.5}[\lF(a)] \land \lP_{>0.99}[\lF(b)]$ (from Table~\ref{tab:correct-incorrect}) took 24.05 seconds for PTS, and 0.05 and 0.01 seconds for GBE and BSC, respectively. 
The heuristics enhance efficiency by reducing the search space, particularly for larger formulas.  
The fifth column of Table~\ref{tab:correct-incorrect} compares the considered formula for PTS with the total space up to the size of the learned formula.
For the same example, the search was reduced to just 7\% of the possible \LTL{} space via heuristics.

\subsubsection{Learning from variants of probabilistic models}
In this experiment, we compare two implementations of the probabilistic secret-sharing protocol EGL~\cite{DBLP:journals/cacm/EvenGL85,NS06}, where two parties, $A$ and $B$, share $2P$ secrets (2-length bit-vectors) over several rounds.
The two implementations we consider, EGL1\textsubscript{$P$} and EGL2\textsubscript{$P$}, are parametrized by the number of secrets $P=1,\dots,7$.
The key difference in the variants is in the sharing order: in EGL1\textsubscript{$P$}, each party sequentially shares all their $i^{\text{th}}$ bits, while in EGL2\textsubscript{$P$}, they alternately share half of their $i^{\text{th}}$ bits.
We provide further details on the implementation in Appendix~\ref{app:egl-exp}.

To apply \tool{}, we treat EGL1\textsubscript{$P$} as positive, EGL2\textsubscript{$P$} as negative, $K=6$ and $\prop=\{\knowA,\knowB\}$, where $\knowA$ ($\knowB$ resp.) represents A (B resp.) knows of B's (A's resp.) secrets.

For $P=1,2,3,4,5$, \tool{} inferred the formula $\lP_{>p}[\neg\knowA \lU \knowB]$ with progressively increasing thresholds $p=0.88, 0.90, 0.94, 0.97$ with running times 0.37, 0.38, 0.68, 1.31 seconds, respectively.
The inferred formula indicates that, in EGL1\textsubscript{$P$}, the probability that $B$ knows $A$'s secret before $A$ knows $B$'s is higher as compared to EGL2\textsubscript{$P$}.

For $P>5$, however, \tool{} did not infer any formula.
This indicates that no \PLTL{} formula with parameters $K\leq 6$, $d\leq 2$ and $\delta>0.05$ distinguishes the variants when a higher number of secrets are shared, based on our exhaustive search (Theorem~\ref{thm:guarantees}).
Overall, \tool{} could identify key probabilistic temporal differences between variants of probabilistic models, or confirm their absence.

\section{Conclusion}

We focused on the automatic learning of temporal behaviour in stochastic systems. 
Specifically, we considered the passive learning problem of learning concise probabilistic \LTL{} (\PLTL{}) formulas that distinguish between positive and negative Markov chains. 
Our novel learning algorithm combines grammar-based enumeration with probabilistic model checking, enhanced by search heuristics. 
We demonstrated the ability of our approach in inferring temporal specifications in both reinforcement learning and modelling applications.

In the future, we plan to integrate our algorithm into other learning frameworks, such as active learning~\cite{CamachoM19} and learning from positive examples~\cite{ltl-from-positive-only}. Moreover, we aim to extend our approach to multi-agent systems~\cite{DBLP:conf/ijcai/BoggessK023}.

\subsubsection{Acknowledgments.}
This project received funding from the ERC under the European Union’s
Horizon 2020 research and innovation programme (grant agreement No.834115, FUN2MODEL).
The work was done while YP visited the University of Oxford for an internship.

\bibliographystyle{named}
\bibliography{bib}

\begin{thebibliography}{}

\bibitem[\protect\citeauthoryear{Baier and Katoen}{2008}]{model-checking-book}
Christel Baier and Joost{-}Pieter Katoen.
\newblock {\em Principles of model checking}.
\newblock {MIT} Press, 2008.

\bibitem[\protect\citeauthoryear{Bartocci \bgroup \em et al.\egroup }{2014}]{DBLP:conf/formats/BartocciBS14}
Ezio Bartocci, Luca Bortolussi, and Guido Sanguinetti.
\newblock Data-driven statistical learning of temporal logic properties.
\newblock In {\em {FORMATS}}, volume 8711 of {\em Lecture Notes in Computer Science}, pages 23--37. Springer, 2014.

\bibitem[\protect\citeauthoryear{Bartocci \bgroup \em et al.\egroup }{2022}]{DBLP:journals/iandc/BartocciMNN22}
Ezio Bartocci, Cristinel Mateis, Eleonora Nesterini, and Dejan Nickovic.
\newblock Survey on mining signal temporal logic specifications.
\newblock {\em Inf. Comput.}, 289(Part):104957, 2022.

\bibitem[\protect\citeauthoryear{Bj{\o}rner and Havelund}{2014}]{BjornerH14}
Dines Bj{\o}rner and Klaus Havelund.
\newblock 40 years of formal methods - some obstacles and some possibilities?
\newblock In {\em {FM}}, volume 8442 of {\em Lecture Notes in Computer Science}, pages 42--61. Springer, 2014.

\bibitem[\protect\citeauthoryear{Boggess \bgroup \em et al.\egroup }{2023}]{DBLP:conf/ijcai/BoggessK023}
Kayla Boggess, Sarit Kraus, and Lu~Feng.
\newblock Explainable multi-agent reinforcement learning for temporal queries.
\newblock In {\em Proceedings of the Thirty-Second International Joint Conference on Artificial Intelligence, {IJCAI} 2023, 19th-25th August 2023, Macao, SAR, China}, pages 55--63. ijcai.org, 2023.

\bibitem[\protect\citeauthoryear{Bombara \bgroup \em et al.\egroup }{2016}]{dtmethod}
Giuseppe Bombara, Cristian~Ioan Vasile, Francisco Penedo, Hirotoshi Yasuoka, and Calin Belta.
\newblock A decision tree approach to data classification using signal temporal logic.
\newblock In {\em Proceedings of the 19th International Conference on Hybrid Systems: Computation and Control}, HSCC '16, page 1–10, New York, NY, USA, 2016. Association for Computing Machinery.

\bibitem[\protect\citeauthoryear{Bordais \bgroup \em et al.\egroup }{2024}]{DBLP:conf/fm/BordaisNR24}
Benjamin Bordais, Daniel Neider, and Rajarshi Roy.
\newblock Learning branching-time properties in {CTL} and {ATL} via constraint solving.
\newblock In Andr{\'{e}} Platzer, Kristin~Yvonne Rozier, Matteo Pradella, and Matteo Rossi, editors, {\em Formal Methods - 26th International Symposium, {FM} 2024, Milan, Italy, September 9-13, 2024, Proceedings, Part {I}}, volume 14933 of {\em Lecture Notes in Computer Science}, pages 304--323. Springer, 2024.

\bibitem[\protect\citeauthoryear{Bozkurt \bgroup \em et al.\egroup }{2020}]{DBLP:conf/icra/Bozkurt0ZP20}
Alper~Kamil Bozkurt, Yu~Wang, Michael~M. Zavlanos, and Miroslav Pajic.
\newblock Control synthesis from linear temporal logic specifications using model-free reinforcement learning.
\newblock In {\em 2020 {IEEE} International Conference on Robotics and Automation, {ICRA} 2020, Paris, France, May 31 - August 31, 2020}, pages 10349--10355. {IEEE}, 2020.

\bibitem[\protect\citeauthoryear{Brockman \bgroup \em et al.\egroup }{2016}]{DBLP:journals/corr/BrockmanCPSSTZ16}
Greg Brockman, Vicki Cheung, Ludwig Pettersson, Jonas Schneider, John Schulman, Jie Tang, and Wojciech Zaremba.
\newblock Openai gym.
\newblock {\em CoRR}, abs/1606.01540, 2016.

\bibitem[\protect\citeauthoryear{Brockman}{2016}]{brockman2016openai}
G~Brockman.
\newblock Openai gym.
\newblock {\em arXiv preprint arXiv:1606.01540}, 2016.

\bibitem[\protect\citeauthoryear{Camacho and McIlraith}{2019}]{CamachoM19}
Alberto Camacho and Sheila~A. McIlraith.
\newblock Learning interpretable models expressed in linear temporal logic.
\newblock In {\em {ICAPS}}, pages 621--630. {AAAI} Press, 2019.

\bibitem[\protect\citeauthoryear{Camacho \bgroup \em et al.\egroup }{2019}]{DBLP:conf/ijcai/CamachoIKVM19}
Alberto Camacho, Rodrigo~Toro Icarte, Toryn~Q. Klassen, Richard~Anthony Valenzano, and Sheila~A. McIlraith.
\newblock {LTL} and beyond: Formal languages for reward function specification in reinforcement learning.
\newblock In Sarit Kraus, editor, {\em Proceedings of the Twenty-Eighth International Joint Conference on Artificial Intelligence, {IJCAI} 2019, Macao, China, August 10-16, 2019}, pages 6065--6073. ijcai.org, 2019.

\bibitem[\protect\citeauthoryear{Chiariello}{2024}]{DBLP:conf/time/Chiariello24}
Francesco Chiariello.
\newblock Learning temporal properties from event logs via sequential analysis.
\newblock In {\em {TIME}}, volume 318 of {\em LIPIcs}, pages 14:1--14:14. Schloss Dagstuhl - Leibniz-Zentrum f{\"{u}}r Informatik, 2024.

\bibitem[\protect\citeauthoryear{Danesh \bgroup \em et al.\egroup }{2021}]{DBLP:conf/icml/DaneshKFK21}
Mohamad~H. Danesh, Anurag Koul, Alan Fern, and Saeed Khorram.
\newblock Re-understanding finite-state representations of recurrent policy networks.
\newblock In {\em {ICML}}, volume 139 of {\em Proceedings of Machine Learning Research}, pages 2388--2397. {PMLR}, 2021.

\bibitem[\protect\citeauthoryear{Duret-Lutz \bgroup \em et al.\egroup }{2022}]{duret.22.cav}
Alexandre Duret-Lutz, Etienne Renault, Maximilien Colange, Florian Renkin, Alexandre~Gbaguidi Aisse, Philipp Schlehuber-Caissier, Thomas Medioni, Antoine Martin, J{\'e}r{\^o}me Dubois, Cl{\'e}ment Gillard, and Henrich Lauko.
\newblock From {S}pot 2.0 to {S}pot 2.10: What's new?
\newblock In {\em Proceedings of the 34th International Conference on Computer Aided Verification (CAV'22)}, volume 13372 of {\em Lecture Notes in Computer Science}, pages 174--187. Springer, August 2022.

\bibitem[\protect\citeauthoryear{Duret-Lutz}{2024}]{spotrewriterules}
Alexandre Duret-Lutz.
\newblock Spot’s temporal logic formulas, 2024.
\newblock Accessed: 02-01-2025.

\bibitem[\protect\citeauthoryear{Dwyer \bgroup \em et al.\egroup }{1998}]{DBLP:conf/fmsp/DwyerAC98}
Matthew~B. Dwyer, George~S. Avrunin, and James~C. Corbett.
\newblock Property specification patterns for finite-state verification.
\newblock In {\em {FMSP}}, pages 7--15. {ACM}, 1998.

\bibitem[\protect\citeauthoryear{Even \bgroup \em et al.\egroup }{1985}]{DBLP:journals/cacm/EvenGL85}
Shimon Even, Oded Goldreich, and Abraham Lempel.
\newblock A randomized protocol for signing contracts.
\newblock {\em Commun. {ACM}}, 28(6):637--647, 1985.

\bibitem[\protect\citeauthoryear{Flajolet and Sedgewick}{2009}]{DBLP:books/daglib/0023751}
Philippe Flajolet and Robert Sedgewick.
\newblock {\em Analytic Combinatorics}.
\newblock Cambridge University Press, 2009.

\bibitem[\protect\citeauthoryear{Fuggitti and Chakraborti}{2023}]{DBLP:conf/aaai/FuggittiC23}
Francesco Fuggitti and Tathagata Chakraborti.
\newblock {NL2LTL} - a python package for converting natural language {(NL)} instructions to linear temporal logic {(LTL)} formulas.
\newblock In {\em {AAAI}}, pages 16428--16430. {AAAI} Press, 2023.

\bibitem[\protect\citeauthoryear{Gold}{1978}]{DBLP:journals/iandc/Gold78}
E.~Mark Gold.
\newblock Complexity of automaton identification from given data.
\newblock {\em Inf. Control.}, 37(3):302--320, 1978.

\bibitem[\protect\citeauthoryear{Hasanbeig \bgroup \em et al.\egroup }{2019}]{DBLP:conf/cdc/HasanbeigKAKPL19}
Mohammadhosein Hasanbeig, Yiannis Kantaros, Alessandro Abate, Daniel Kroening, George~J. Pappas, and Insup Lee.
\newblock Reinforcement learning for temporal logic control synthesis with probabilistic satisfaction guarantees.
\newblock In {\em 58th {IEEE} Conference on Decision and Control, {CDC} 2019, Nice, France, December 11-13, 2019}, pages 5338--5343. {IEEE}, 2019.

\bibitem[\protect\citeauthoryear{Kemeny \bgroup \em et al.\egroup }{1976}]{KSK76}
J.~Kemeny, J.~Snell, and A.~Knapp.
\newblock {\em Denumerable {M}arkov Chains}.
\newblock Springer-Verlag, 2nd edition, 1976.

\bibitem[\protect\citeauthoryear{Kwiatkowska \bgroup \em et al.\egroup }{2011}]{DBLP:conf/cav/KwiatkowskaNP11}
Marta~Z. Kwiatkowska, Gethin Norman, and David Parker.
\newblock {PRISM} 4.0: Verification of probabilistic real-time systems.
\newblock In Ganesh Gopalakrishnan and Shaz Qadeer, editors, {\em Computer Aided Verification - 23rd International Conference, {CAV} 2011, Snowbird, UT, USA, July 14-20, 2011. Proceedings}, volume 6806 of {\em Lecture Notes in Computer Science}, pages 585--591. Springer, 2011.

\bibitem[\protect\citeauthoryear{Li \bgroup \em et al.\egroup }{2017}]{DBLP:conf/iros/LiVB17}
Xiao Li, Cristian~Ioan Vasile, and Calin Belta.
\newblock Reinforcement learning with temporal logic rewards.
\newblock In {\em 2017 {IEEE/RSJ} International Conference on Intelligent Robots and Systems, {IROS} 2017, Vancouver, BC, Canada, September 24-28, 2017}, pages 3834--3839. {IEEE}, 2017.

\bibitem[\protect\citeauthoryear{Luo \bgroup \em et al.\egroup }{2022}]{DBLP:conf/aaai/LuoLDWPZ22}
Weilin Luo, Pingjia Liang, Jianfeng Du, Hai Wan, Bo~Peng, and Delong Zhang.
\newblock Bridging ltlf inference to {GNN} inference for learning ltlf formulae.
\newblock In {\em {AAAI}}, pages 9849--9857. {AAAI} Press, 2022.

\bibitem[\protect\citeauthoryear{Milani \bgroup \em et al.\egroup }{2024}]{DBLP:journals/csur/MilaniTVF24}
Stephanie Milani, Nicholay Topin, Manuela Veloso, and Fei Fang.
\newblock Explainable reinforcement learning: {A} survey and comparative review.
\newblock {\em {ACM} Comput. Surv.}, 56(7):168:1--168:36, 2024.

\bibitem[\protect\citeauthoryear{Mohammadinejad \bgroup \em et al.\egroup }{2020}]{MohammadinejadD20}
Sara Mohammadinejad, Jyotirmoy~V. Deshmukh, Aniruddh~Gopinath Puranic, Marcell Vazquez{-}Chanlatte, and Alexandre Donz{\'{e}}.
\newblock Interpretable classification of time-series data using efficient enumerative techniques.
\newblock In {\em {HSCC} '20: 23rd {ACM} International Conference on Hybrid Systems: Computation and Control, Sydney, New South Wales, Australia, April 21-24, 2020}, pages 9:1--9:10. {ACM}, 2020.

\bibitem[\protect\citeauthoryear{Neider and Gavran}{2018}]{flie}
Daniel Neider and Ivan Gavran.
\newblock Learning linear temporal properties.
\newblock In Nikolaj~S. Bj{\o}rner and Arie Gurfinkel, editors, {\em 2018 Formal Methods in Computer Aided Design, {FMCAD} 2018, Austin, TX, USA, October 30 - November 2, 2018}, pages 1--10. {IEEE}, 2018.

\bibitem[\protect\citeauthoryear{Neider and Roy}{2025}]{DBLP:conf/birthday/Neider025}
Daniel Neider and Rajarshi Roy.
\newblock What is formal verification without specifications? {A} survey on mining {LTL} specifications.
\newblock In {\em Principles of Verification {(3)}}, volume 15262 of {\em Lecture Notes in Computer Science}, pages 109--125. Springer, 2025.

\bibitem[\protect\citeauthoryear{Nenzi \bgroup \em et al.\egroup }{2018}]{NenziSBB18}
Laura Nenzi, Simone Silvetti, Ezio Bartocci, and Luca Bortolussi.
\newblock A robust genetic algorithm for learning temporal specifications from data.
\newblock In {\em {QEST}}, volume 11024 of {\em Lecture Notes in Computer Science}, pages 323--338. Springer, 2018.

\bibitem[\protect\citeauthoryear{Norman and Shmatikov}{2006}]{NS06}
G.~Norman and V.~Shmatikov.
\newblock Analysis of probabilistic contract signing.
\newblock {\em Journal of Computer Security}, 14(6):561--589, 2006.

\bibitem[\protect\citeauthoryear{Pnueli}{1977}]{DBLP:conf/focs/Pnueli77}
Amir Pnueli.
\newblock The temporal logic of programs.
\newblock In {\em 18th Annual Symposium on Foundations of Computer Science, Providence, Rhode Island, USA, 31 October - 1 November 1977}, pages 46--57. {IEEE} Computer Society, 1977.

\bibitem[\protect\citeauthoryear{Pommellet \bgroup \em et al.\egroup }{2024}]{DBLP:conf/ijcar/PommelletSS24}
Adrien Pommellet, Daniel Stan, and Simon Scatton.
\newblock Sat-based learning of computation tree logic.
\newblock In Christoph Benzm{\"{u}}ller, Marijn J.~H. Heule, and Renate~A. Schmidt, editors, {\em Automated Reasoning - 12th International Joint Conference, {IJCAR} 2024, Nancy, France, July 3-6, 2024, Proceedings, Part {I}}, volume 14739 of {\em Lecture Notes in Computer Science}, pages 366--385. Springer, 2024.

\bibitem[\protect\citeauthoryear{Raha \bgroup \em et al.\egroup }{2022}]{scarlet}
Ritam Raha, Rajarshi Roy, Nathana{\"e}l Fijalkow, and Daniel Neider.
\newblock Scalable anytime algorithms for learning fragments of linear temporal logic.
\newblock In Dana Fisman and Grigore Rosu, editors, {\em Tools and Algorithms for the Construction and Analysis of Systems}, pages 263--280, Cham, 2022. Springer International Publishing.

\bibitem[\protect\citeauthoryear{Roy \bgroup \em et al.\egroup }{2020}]{DBLP:conf/ijcai/0002FN20}
Rajarshi Roy, Dana Fisman, and Daniel Neider.
\newblock Learning interpretable models in the property specification language.
\newblock In {\em {IJCAI}}, pages 2213--2219. ijcai.org, 2020.

\bibitem[\protect\citeauthoryear{Roy \bgroup \em et al.\egroup }{2022}]{ltl-from-positive-only}
Rajarshi Roy, Jean{-}Rapha{\"{e}}l Gaglione, Nasim Baharisangari, Daniel Neider, Zhe Xu, and Ufuk Topcu.
\newblock Learning interpretable temporal properties from positive examples only.
\newblock {\em CoRR}, abs/2209.02650, 2022.

\bibitem[\protect\citeauthoryear{Rozier}{2016}]{Rozier16}
Kristin~Yvonne Rozier.
\newblock Specification: The biggest bottleneck in formal methods and autonomy.
\newblock In {\em {VSTTE}}, volume 9971 of {\em Lecture Notes in Computer Science}, pages 8--26, 2016.

\bibitem[\protect\citeauthoryear{Shao and Kwiatkowska}{2023}]{DBLP:conf/ijcai/ShaoK23}
Daqian Shao and Marta Kwiatkowska.
\newblock Sample efficient model-free reinforcement learning from {LTL} specifications with optimality guarantees.
\newblock In {\em Proceedings of the Thirty-Second International Joint Conference on Artificial Intelligence, {IJCAI} 2023, 19th-25th August 2023, Macao, SAR, China}, pages 4180--4189. ijcai.org, 2023.

\bibitem[\protect\citeauthoryear{Topin and Veloso}{2019}]{DBLP:conf/aaai/TopinV19}
Nicholay Topin and Manuela Veloso.
\newblock Generation of policy-level explanations for reinforcement learning.
\newblock In {\em {AAAI}}, pages 2514--2521. {AAAI} Press, 2019.

\bibitem[\protect\citeauthoryear{Valizadeh \bgroup \em et al.\egroup }{2024}]{LTL-GPU}
Mojtaba Valizadeh, Nathana{\"e}l Fijalkow, and Martin Berger.
\newblock Ltl learning on gpus.
\newblock In Arie Gurfinkel and Vijay Ganesh, editors, {\em Computer Aided Verification}, pages 209--231, Cham, 2024. Springer Nature Switzerland.

\bibitem[\protect\citeauthoryear{Vardi}{1985}]{Var85}
M.~Vardi.
\newblock Automatic verification of probabilistic concurrent finite state programs.
\newblock In {\em Proc. 26th Annual Symposium on Foundations of Computer Science (FOCS'85)}, pages 327--338. IEEE Computer Society Press, 1985.

\bibitem[\protect\citeauthoryear{Wan \bgroup \em et al.\egroup }{2024}]{DBLP:conf/aaai/WanLDLYP24}
Hai Wan, Pingjia Liang, Jianfeng Du, Weilin Luo, Rongzhen Ye, and Bo~Peng.
\newblock End-to-end learning of ltlf formulae by faithful ltlf encoding.
\newblock In {\em {AAAI}}, pages 9071--9079. {AAAI} Press, 2024.

\end{thebibliography}

\appendix
\section{Additional Proofs}
\label{app:proofs}
\subsubsection{Proof of Lemma~\ref{lem:gbe-completeness}}
    We prove the lemma by induction on the size $N'$.

For the \emph{base case} $N'=1$,
$\formulalist_1^0 = \literals$, $\formulalist_1^d=\emptyset$ for all $d \leq D$, which constitute the only possible formulas for this size.

For the inductive step, assume the claim holds for $\mathcal{L} := \bigcup_{n \leq N'} \formulalist_n$.
We now show that $\formulalist_{N'+1}$ computed by GBE includes all formulas of size $N'+1$ and depth $\leq D$ that are semantically distinct from those in $\mathcal{L}$.

To this end, GBE systematically constructs all formulas of size $N'+1$ by applying LTL operators to formulas in $\mathcal{L}$ and discards a formula only if it is identified as redundant by the corresponding heuristic.

Specifically, GBE applies unary operators (Line 4) as $\psi := \circ \varphi$ for $\circ \in \{\lF, \lG, \lX\}$ and $\varphi \in \mathcal{L}$, and discards $\psi$ only if temporal simplification applies, i.e., there exists $\psi' \in \mathcal{L}$ such that $\psi \equiv \psi'$.
It then applies binary operators (Lines 9, 13) as $\psi := \varphi \circ \varphi'$ for $\circ \in \{\lU, \lor, \land\}$ and $\varphi, \varphi' \in \mathcal{L}$, discarding $\psi$ only if Boolean simplification applies, i.e., $\psi \equiv \varphi$ or $\psi \equiv \lfalse$.
Thus, $\formulalist_{N'+1}$ contains all formulas of size $N'+1$ and depth $\leq D$ that are not semantically equivalent to any formula in $\mathcal{L}$, completing the inductive step.

\subsubsection{Proof of Lemma~\ref{lem:inconsistency-removal}}
    \begin{proof}
    To prove this lemma, we first note a general property of LTL for probabilistic models:
    for some $s\in S$, $\probvector^{\markov,\psi}(s)=0$ if and only if $\PM^{\markov}_s(\{\pi\in \paths^{\markov}(s)~|~\pi\models\psi\})=0$ if and only if $\pi\not\models\psi$ for any $\pi\in \Pi^\markov(s)$.
    
    Based on this general property, $\probvector^{\markov,\varphi}\equiv \mathbf{0}$ if and only if $\pi\not\models\varphi$ for any path $\pi\in\Pi^{\markov}$ starting from any state in $\markov$.
    
    Now, assume $\pi\models\lX(\varphi)$ for some path $\pi\in\Pi^{\markov}$.
    Based on the semantics of $\lX$, {$\pi[1:]\models\varphi$} for some $i\in\nat$, which leads to a contradiction as $\varphi$ does not hold on any path.
    A similar contradiction works for $\psi = \{\lF(\varphi), \lG(\varphi), \lX(\varphi),\varphi' \lU \varphi,\varphi' \land \varphi\}$ and consequently, $\probvector^{\markov,\psi}\equiv \mathbf{0}$. 
    Therefore, $\psi$ is not a useful formula.
    
    For $\psi=\{\varphi\lU\varphi',\varphi\lor\varphi'\}$, we have $\pi\models\psi$ if and only if $\pi\models\varphi'$ for any $\pi\in\Pi^{\markov}$.
    Therefore, $\PM^{\markov}_s(\{\pi\in \paths^{\markov}(s)~|~\pi\models\psi\})=\PM^{\markov}_s(\{\pi\in \paths^{\markov}(s)~|~\pi\models\varphi'\})$ for any $s\in S$ and consequently, $\probvector^{\markov,\psi}\equiv \probvector^{\markov,\varphi'}$.
    Thus, simply $\varphi'$ can be used instead of $\psi$, which is a smaller formula.
\end{proof}

\section{Implementation Details}
\label{app:implementation}

We implemented the learning algorithm \tool{} in Python 3.12. The full source code, along with the datasets used can be found in our Github project\footnote{\url{https://github.com/rajarshi008/PriTL}}. 
The main dependencies of \tool{} include SPOT version 2.12.2\footnote{https://spot.lre.epita.fr/}, PRISM version 4.6, and other standard PyPI packages. 
We modified the PRISM output to facilitate parsing and extraction of probability vectors, as it is invoked via Python; the modified version is available in our fork\footnote{\url{https://github.com/yashpote/prism}}.

To reduce the overhead of repeatedly invoking PRISM, we query PRISM in batches, consisting of all formulas $\formulalist^d_n$ for all sizes $n$ and $d$. 
Additionally, we utilise the NailGun mode, accessed via the ngprism binary, to eliminate the JVM startup time for each PRISM call, resulting in a significant speedup. 
Finally, we use the flags \texttt{--maxiters 1000000} and \texttt{--exportvector} to set a very high iteration limit to ensure convergence and export the probability vectors, respectively.

\section{Experimental Details}
\label{app:experimental}
We discuss the detailed experimental setup for each of the case studies presented in the main paper.

\subsection{Learning from strategies in stochastic environment}
\label{app:strategy-exp}
\subsubsection{Frozen Lake Environment}

We use the same Frozen Lake environment~\cite{brockman2016openai} as employed in~\cite{DBLP:conf/ijcai/ShaoK23}. The environment is depicted in Figure~\ref{fig:frozen_env}. Blue states represent the frozen lake, where the agent has a 1/3 probability of moving in the intended direction and a 1/3 probability of moving sideways (left or right). The white states labelled $h$ are holes, while the states labelled $a$ and $b$ represent lake camps.

\begin{figure}[h]
\centering
\includegraphics[width=0.35\textwidth]{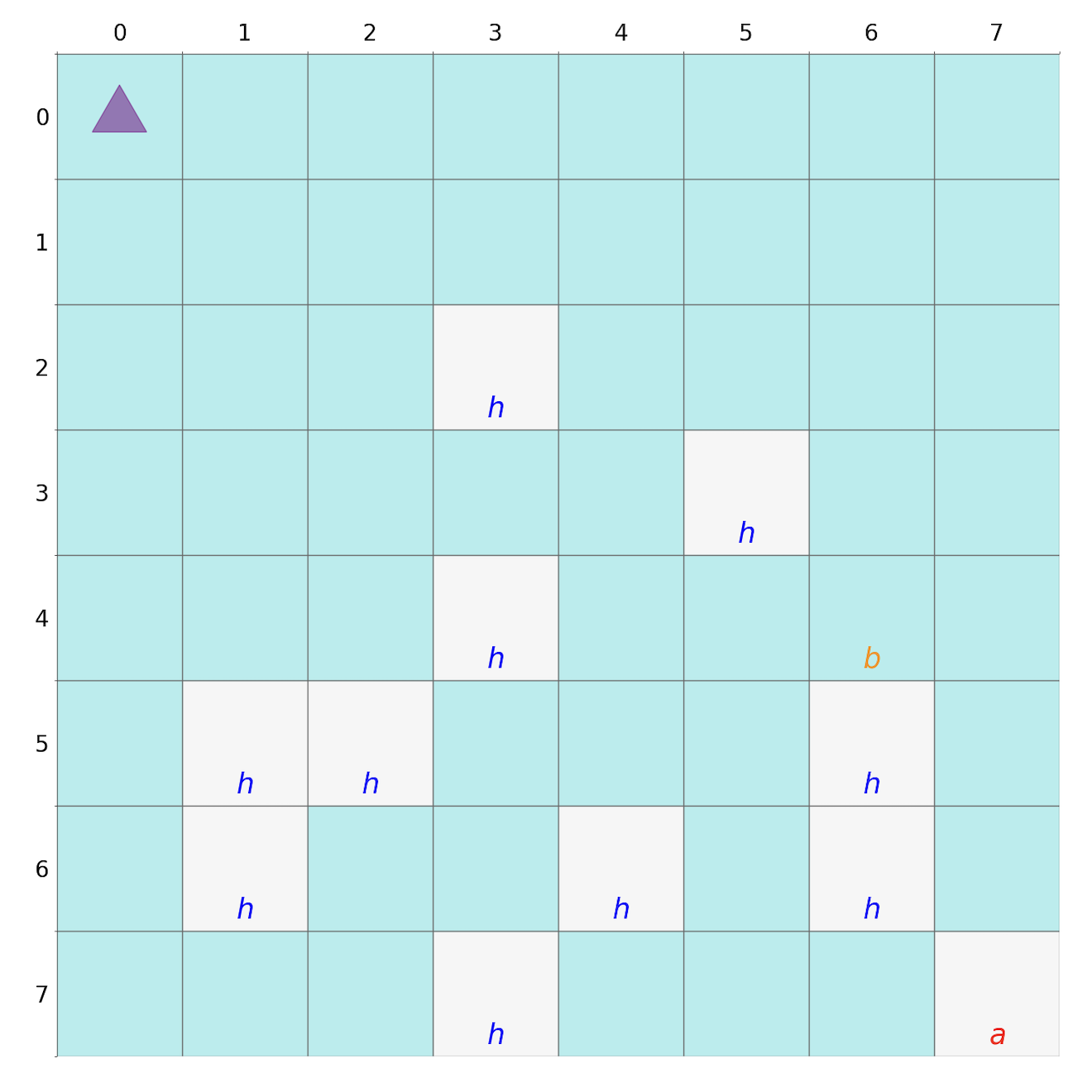}
\caption{The MDP environment for the frozen lake task. Blue represents ice, $h$ are holes, $a$ and $b$ are lake camps, and the purple triangle is the start.}
\label{fig:frozen_env}
\end{figure}

\subsubsection{Q-learning algorithms for LTL tasks}
We utilize the Q-learning algorithms introduced in~\cite{DBLP:conf/ijcai/ShaoK23} to generate strategy DTMCs for the LTL tasks. 
For this, we rely on the authors' GitHub repository\footnote{https://github.com/shaodaqian/rl-from-ltl}, which provides a Python implementation of several algorithms with varying convergence rates. 
The repository also includes a function that automatically constructs strategy DTMCs from the MDP and the strategy in the PRISM language, making it easier for us to obtain the sample of DTMCs.

\subsubsection{Benchmark generation for application (i)}
To generate optimal strategies for the tasks listed in Table~\ref{tab:correct-incorrect}, we use the CF+KC algorithm from~\cite{DBLP:conf/ijcai/ShaoK23}. This algorithm is shown to have a fast convergence rate and is able to produce the most stable strategy DTMCs, which is why we chose it to generate this benchmark set.

We set the number of episodes for CF+KC to be 5000 for all the tasks, employing the default settings from the repository for other parameters.
For each LTL task $\varphi$, we check the probability $v^{\markov,\varphi}_I$ (computed internally by CF+KC) of satisfying the strategy DTMC every 10 episodes. 
We record the first 10 DTMCs, where $v^{\markov,\varphi}_I > 0.99$, as this indicates that the algorithm has converged. 
If two LTL formulas are used to generate negative strategies, we save 5 DTMCs for each formula to ensure a balanced representation.

On this dataset, we run the learning algorithm with propositions $\prop = \{a,b,h\}$, probabilistic difference $\delta=0.05$, and maximum Boolean combination $L=10$.

\subsubsection{Benchmark generation for application (ii)}
In this case, we use the CF algorithm from~\cite{DBLP:conf/ijcai/ShaoK23}, because it has less stability than CF+KC, and often produces suboptimal strategies in early episodes.

We set the number of episodes for KC to be 5000 for all the tasks, employing the default settings from the repository for other parameters.
For each LTL task $\varphi$, we check the probability $v^{\markov,\varphi}_I$ (computed internally by KC) of satisfying the strategy DTMC $\markov$ in every episode.
If $v^{\markov,\varphi}_I\geq0.95$, $M$ is recorded as a positive DTMC, while if $0.5\leq v^{\markov,\varphi}_I<0.9$, $M$ is recorded as a negative DTMC.
We stop the algorithm as soon as we have 30 positive and 30 negative DTMCs for each LTL formula.
The rationale for collecting more DTMCs compared to the previous experiment is that suboptimal strategies may exhibit considerable variance in their temporal behaviour due to the inherent randomness in Q-learning.

\subsection{Learning from EGL protocols}
\label{app:egl-exp}
We elaborate on the description of EGL protocols. 
We use the implementation of the protocols from the case studies presented in PRISM website\footnote{https://www.prismmodelchecker.org/tutorial/egl.php}.
We briefly describe the setting of the protocol. 
\begin{itemize}
\item $A$ and $B$ holds $2P$ secrets $a_1,\dots,a_{2P}$, and $b_1,\dots,b_{2P}$, respectively.
\item all the secrets $a_i$, $b_i$ are a binary string of length 2.
\item the secrets are partitioned into pairs: e.g. $\{ (a_i, a_{P+i})~|~i=1,\dots,P \}$
\item we say A is committed if B knows one of A’s pairs, which we denote as $\knowB$; 
similarly, we say B is committed if A knows one of B’s pairs, which we denote as $\knowA$.
\end{itemize}

In the first part of the secret sharing process, parties probabilistically share some bits of secrets using the \emph{1-out-of-2 oblivious transfer protocol}. Formally, this protocol $OT(S,R,x,y)$ is defined as follows:
\begin{itemize}
\item the sender $S$ sends $x$ and $y$ to receiver $R$
\item $R$ receives $x$ with probability 0.5 otherwise receives $y$
\item $S$ does not know which one $R$ receives
if $S$ cheats then $R$ can detect this with probability 0.5.
\end{itemize}

After this, the parties exchange the remaining bits of their secrets in a specific order, which differs between EGL1 and EGL2. 
Both protocols are described in Algorithm~\ref{alg:egl1} and Algorithm~\ref{alg:egl2}. 
The key difference is in how the bits are shared.
In EGL1, party A shares all of its first bits with party B, and then party B shares all of its first bits with party A. This process is repeated for the second bits. In EGL2, however, party A shares half of its first bits with party B, then party B shares their half with party A, and this continues iteratively.

The propositions used for the learning algorithm are as follows:
\begin{align*}
\knowB &= (a_0 \land a_P) \lor \dots \lor (a_{P+1} \land a_{2P}) \\
\knowA &= (b_0 \land b_P) \lor \dots \lor (b_{P+1} \land b_{2P}) 
\end{align*}
where $\knowB$ denotes that $B$ knows at least one of the pairs of $A$'s secret, meaning $A$ is committed.
Similarly, $\knowA$ denotes that $A$ knows at least one of the pairs of $B$'s secret, meaning $B$ is committed.

\begin{algorithm}[t]
    \caption{EGL1}\label{alg:egl1}
    \begin{algorithmic}[1]
        \FOR{$i=1,\dots,n$}
            \STATE $OT(A,B,a_i,a_{P+i})$
            \STATE $OT(B,A,b_i,b_{P+i})$
        \ENDFOR
        \FOR{$i=1,2$}
        \STATE \textbf{for} $j=1,\dots,2P$, A transmits bit $i$ of secret $a_j$ to B
        \STATE \textbf{for} $j=1,\dots,2P$, B transmits bit $i$ of secret $b_j$ to A
        \ENDFOR
    \end{algorithmic}
\end{algorithm}

\begin{algorithm}[t]
    \caption{EGL2}\label{alg:egl2}
    \begin{algorithmic}[1]
        \FOR{$i=1,\dots,n$}
            \STATE $OT(A,B,a_i,a_{P+i})$
            \STATE $OT(B,A,b_i,b_{P+i})$
        \ENDFOR
        \FOR{$i=1,2$}
        \STATE \textbf{for} $j=1,\dots,P$, A transmits bit $i$ of secret $a_j$ to B
        \STATE \textbf{for} $j=1,\dots,P$, B transmits bit $i$ of secret $b_j$ to A
        \STATE \textbf{for} $j=P+1,\dots,2P$, A transmits bit $i$ of secret $a_j$ to B
        \STATE \textbf{for} $j=P+1,\dots,2P$, B transmits bit $i$ of secret $b_j$ to A
        \ENDFOR
    \end{algorithmic}
\end{algorithm}
\end{document}